\documentclass{article}

\usepackage[active]{srcltx} 
\usepackage{amsmath,amsxtra,amsthm,amssymb}
\usepackage{latexsym,array,verbatim,fancybox,multirow}
\usepackage{epsfig,graphicx,url,setspace,cite}
\usepackage[ruled,vlined,linesnumbered]{algorithm2e}

\newcommand{\norm}[1]{\left\Vert#1\right\Vert}

\newcommand{\Real}{\mathbb R}

\newcommand{\mc}{\mathcal}

\newcommand{\D}{\mathcal{D}}
\newcommand{\x}{\mathbf{x}}

\renewcommand{\a}{\alpha}
\renewcommand{\b}{\beta}
\renewcommand{\l}{\lambda}

\newtheorem{thm}{Theorem}[section]

\newtheorem{prop}{Proposition}
\newtheorem{assm}{Assumption}

\newtheorem{defn}{Definition}

\title{Incentive Games and Mechanisms \\ for Risk Management}

\author{Tansu Alpcan \\
             Technical University of Berlin, \\
              Deutsche Telekom Laboratories,  \\
              10587, Berlin, Germany \\
              Web: \textit{www.tansu.alpcan.org} \\
	      E-mail: \textit{alpcan@sec.t-labs.tu-berlin.de}   
}

\begin{document}
\maketitle

\begin{abstract}
Incentives play an important role in (security and IT) risk management of a large-scale organization
with multiple autonomous divisions. This paper presents an incentive mechanism design framework
for risk management based on a game-theoretic approach. The risk manager acts 
as a mechanism designer providing rules and incentive factors such as assistance or subsidies 
to divisions or units, which are modeled as selfish players of a strategic (noncooperative) game. 
Based on this model, incentive mechanisms with various objectives are developed that satisfy 
efficiency, preference-compatibility, and strategy-proofness criteria. 
In addition, iterative and distributed algorithms are presented, which can be implemented under information limitations such as the risk manager not knowing the individual units' preferences. An example scenario illustrates the framework and results numerically. The incentive mechanism design approach presented is useful for not only deriving guidelines but also developing computer-assistance systems for large-scale risk management.

\end{abstract}

\section{Introduction} \label{sec:intro}

Security risk management is a multi-disciplinary field with both \textbf{technical and organizational dimensions}. On the technical side, complex and networked systems play an increasingly important role in daily business processes. Hence, system failures and security problems have direct consequences for organizations both monetarily and in terms of productivity \cite{moore_red}. It is therefore a necessity
for any modern organization to develop and deploy technical solutions for improving robustness of these complex information technology (IT) systems with respect to failures (e.g. in the form of redundancies)  and defending them against security threats (e.g. firewalls and intrusion detection/response systems).

However, even the best and most suitable technical solution will fail to perform adequately if it is not properly deployed
and supported organizationally. In order to be successful in risk management, an organization has to have proper 
information about its business processes and complex technical systems or ``observe'' them as well as be able
to influence their operation or ``control'' them \cite{alpcan-book}. In a large-scale organization these two necessary 
requirements, which may seem easy to satisfy at first glance, pose significant challenges. An important reason behind 
this issue, beside organizational structure, is the underlying incentive mechanisms.

Autonomous yet interdependent divisions or units of a large organization have often \textbf{individual objectives and incentives} that may 
not be as aligned in practice as the headquarters and executives wish. Each such unit may have a different
perspective on risk management which directly affects deployment of technical or organizational solutions. Misaligned incentives
also make observation and control of business and technical processes difficult for risk managers. Considering the complex
interdependencies in today's technology and business, such a misalignment in incentives is not a luxury even a large-scale
organization can effort.

Let us consider an \textbf{example scenario} of an enterprise deploying a new security risk management system
that entails information collection (observation), risk assessment (decision making), and mitigation (control).
In order for its successful operation, each division has to cooperate at each stage of its deployment and operation. 
At the deployment phase, the divisions have to provide accurate information on their business and networked systems.
During the operational phase, each division has to allocate manpower and resources for the proper operation of 
the system. All these can be accomplished only if the division has sufficient incentives for real cooperation. Otherwise,
the risk management system would simply fail as a result of bureaucracy, enterprise politics, and delaying tactics.


\textbf{Game theoretic approaches} have significant potential in addressing the above described issues as well as
in risk analysis, management, and associated decision making \cite{riskbook1,crisis09,icc10jeff}.
The performance of manual and heuristic schemes degrades fast as the scale and complexity of the organization increases. 
Computer assistance in observation, decision making, and control of different risk management aspects is necessary to overcome this problem. Development of such computer-based support schemes, however, require quantitative representations and analysis. Game theoretic and analytical frameworks provide a mathematical abstraction which is useful for generalization of seemingly different problems, combining the existing ad-hoc schemes under a single umbrella, and opening doors to novel solutions. At the same time, such frameworks and the associated scientific methodology leads to streamlining of risk management processes and possibly more transparency as a consequence of increased observability and control \cite{alpcan-book}.

\textbf{Mechanism design} \cite{maskin1,lazarSemret1998,johari1}, which is a field of game theory, has been proposed recently as a way to model, analyze, and address  risk management problems \cite{alpcan-book}. It can be potentially useful especially in developing analytical frameworks for incentive mechanisms. Game theory in general provides a rich set of mathematical tools and models for investigating multi-person strategic decision making  where the players (decision makers) compete for limited and shared resources \cite{basargame,fudenberg}. Mechanism design studies ways of designing rules and structure of games such that their outcome achieve certain objectives. 

In the context of security risk management, the units of an organization can be modeled as players (independent decision makers)
in a \textbf{risk management game} since they share and compete for organizational resources. Each player decides on the allocation of
unit's resources, e.g. in terms of manpower and investments, to assess and mitigate perceived risks. The task of organization's risk 
manager (designer) is then influence the outcome of this game by imposing rules and varying its structure such that a satisfactory
amount of investment is made by each unit. Thus, the designer tries to optimize the risk management process 
from the entire organization's perspective within given resource constraints, e.g. budget.

This paper adopts a \textbf{game-theoretic approach} and presents a framework of incentive mechanism design for security risk management. The analytical framework studied can not only be used to derive guidelines for handling incentives in risk management but also to develop computer-assisted risk management systems. The \textbf{main contributions} of the paper include:
\begin{itemize}
 \item A strategic (noncooperative) game approach for analysis of incentives in (security and IT) risk management.
 \item An analytical incentive mechanism design framework where the designer does not have access to utilities of individual players of the underlying strategic game.
 \item Study of iterative incentive schemes which can be implemented under information limitations and their convergence analysis.
 \item A numerical analysis based on a scenario of a risk management system deployment.
\end{itemize}
A more detailed discussion clarifying these contributions and a comparison with existing literature will be provided in Section~\ref{sec:discussion}.

The rest of the paper is organized as follows. The next section provides an overview of the underlying mechanism design 
and game-theoretic concepts as well as the model adopted in this work. Section~\ref{sec:incentivemech}
presents incentive mechanism design for risk management. Section~\ref{sec:iterative} discusses
iterative incentive mechanisms and related distributed algorithms. An example use case scenario and related numerical analysis is presented in Section~\ref{sec:numerical},
which is followed by a  brief literature review  in Section~\ref{sec:discussion}. The paper concludes with a discussion
and concluding remarks in Section~\ref{sec:conclusion}.


\section{Game and Mechanism Model} \label{sec:model}

Consider an organization with $N$ \textit{autonomous units}, which act as independent decision makers, and a risk manager, 
which oversees the risk management task of the entire organization  (and is often a special organizational unit itself).
This generic organization may be a large-scale multi-national enterprise (divisions versus the risk manager at the headquarters), a government (government agencies versus central executives), or even an international organization (individual countries versus general secretary of the organization).

Adopting a game-theoretic approach, each autonomous unit can be modeled as a player of a 
\textbf{strategic (noncooperative) game} with the set of all players denoted as $\mc A$. The player $i \in \mc A$ independently decides on
its respective decision variable $x_i$, which represents allocation of limited resources such as monetary investments or manpower,
in accordance with own objectives. In majority of cases, the decisions of players affect each other due to constraints of the environment.
Thus, the players share and compete for resources  as part of this strategic game.

The risk manager $\D$, which is also called \textit{designer}\footnote{The terms risk manager and designer as well as (organizational) unit and player will be used interchangeably for the rest of the paper. }
 in the context of mechanism design, focuses on the aggregate outcome
of the strategic game and tries to ensure that the game satisfies some risk management objectives, e.g. information
collection for assessment or deployment of a new risk management solution. Unlike the players, the designer
achieves its objective only by \textit{indirect means such as providing additional incentives to players} in the form of incentive factors and
penalties or imposing rules. It is important to note that the risk manager cannot directly dictate individual actions of players, which is
a realistic assumption that holds for many types of civilian organizations. The interaction between risk manager (designer) and organizational units (players) is depicted in Figure~\ref{fig:mechdesign1}.
\begin{figure}[htbp]
  \centering
  \includegraphics[width=0.7 \columnwidth]{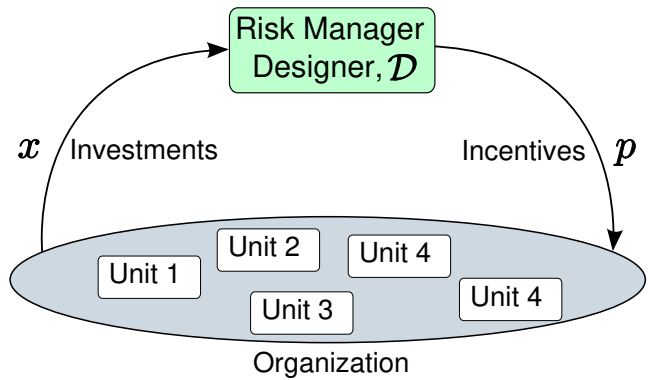}
  \caption{The interaction between the players (autonomous organizational units) of the underlying strategic game and the
  risk manager acting as mechanism designer, who observes players actions (investments) $x$ and provides additional
   incentives $p$.}
\label{fig:mechdesign1}
\end{figure}

The $N$-player strategic game, $\mc G$ is described as follows. Each player $i \in \mc A$ has a respective  
scalar \textbf{decision variable}\footnote{The analysis can be easily extended to multi-dimensional case. However, since
this would complicate the notation and readability without a significant conceptual contribution, 
this paper focuses on scalar decision variables.} $x_i$ such that 
$$x=[x_1,\ldots,x_N] \in \mc X \subset \Real^N, $$ 
where $\mc X$ is the convex, compact, and nonempty decision space of all players. 
The players make their decisions in accordance with their \textbf{preferences} modeled as customary by real valued utility functions 
$$  U_i (x) : \mc X \rightarrow \Real .$$ 
For analytical tractability, the player utility functions are chosen as continuous, differentiable, and strictly concave.
It is important to note here that \textit{players do not reveal their utilities (preferences) to the designer}.
Application of a similar utility function approach to risk management has been discussed in detail in \cite[Chap. 3]{riskbook1}, where the concave utilities are interpreted as ``risk averse''.

While each player gains a utility from its decisions (investments), these resources also have a cost, which can be often
expressed in monetary terms. We assume that that these costs are linear in the allocated resource, 
$ \b_i x_i $, where $\beta_i$ is the individual per unit cost factor. Each player $i$ aims to minimize its 
respective \textbf{cost function}
\begin{equation} \label{e:usercost}
  J_i(x)=  \b_i x_i - U_i (x) - p_i x_i ,
\end{equation}
where the linear term $p_i x_i$ represents the \textit{incentive factor} (or penalty if negative)
provided to the player by the designer $\D$.
Thus, player $i$  solves the optimization problem 
$$ \min_{x_i} J_i (x_i,x_{-i}) ,$$ 
by choosing an appropriate $x_i$ given the decisions of all players denoted by $x_{-i}$ such that $x \in \mc X$.
Formally, strategic game $\mc G$ is defined as:

\begin{defn} \label{def:game}
The strategic (noncooperative) game $\mc G$ is played among the set of selfish players, $\mc A$, of cardinality $N$, on the convex, compact, and non-empty decision space $\mc X \subset \Real^N$, where 
\begin{itemize}
 \item $x=[x_1,\ldots,x_N] \in \mc X$ denotes the actions of players
 \item $  U_i (x) : \mc X \rightarrow \Real$ denotes the utility function of player $i \in \mc A$
 \item $  J_i(x)=  \b_i x_i - U_i (x) - p_i x_i$ denotes the cost function of player $i \in \mc A$ for given parameters $b_i$ and $p_i$ $\forall i$,
\end{itemize}
such that each player $i$  solves its own optimization problem 
$$ \min_{x_i} J_i (x_i,x_{-i}) ,$$ 
by choosing an appropriate $x_i$ given the decisions of all players denoted by $x_{-i}$.
\end{defn}

The \textbf{Nash equilibrium} (NE) is a widely-accepted and useful solution concept in strategic games, where no player has an incentive to deviate from it while others play according to their NE strategies \cite{Nash50,Nash51}. The NE is at the same time the intersection point of players' best responses obtained by solving  their individual optimization problems. The NE of the game $\mc G$ in Definition~\ref{def:game} is formally defined as follows.
\begin{defn} \label{def:ne}
The Nash equilibrium of the game $\mc G$ in Definition~\ref{def:game}, is denoted by the vector $x^*=[x_1^*,\ldots,x_N^*] \in \mc X$ and defined as
$$ x_i^* := \arg \min_{x_i} J_i (x_i, x_{-i}^*)\;\; \;\forall i \in \mc A,$$
where $x_{-i}^*=[x_1^*,\ldots,x_{i-1}^*,x_{i+1}^*,\ldots, x_N^*]$. 
\end{defn}

If some special convexity and compactness conditions are imposed to the game $\mc G$, then it admits a unique NE solution, which simplifies mechanism and algorithm design significantly.  
We refer to the Appendix A.1 as well as \cite{rosen,basargame,tansuphd} for the details and
an extensive analysis.

The risk manager (designer) $\D$ devises an  \textbf{incentive mechanism} $\mc M$, which can be represented by the mapping 
$\mc M: \mc X \rightarrow \Real^N$, and implemented through additional incentives (e.g. subsidies) in player cost functions, $p_i x_i$, above. 
Using incentive mechanism $\mc M$, the designer aims to achieve a certain risk management objective, which can be maximization of 
aggregate player utilities (expected aggregate benefit from risk-related investments) or an independent organizational target that 
depends on participation of all players such as deployment of a new risk management solution. 
These can be modeled using a \textbf{designer objective function} $V$ that quantifies the desirability of an outcome $x$ from 
the designers perspective. Formally, the function $V$ is defined as 
$$ V(x,U(x),p) : \mc X \rightarrow \Real.$$ 
Thus, the global optimization problem of the designer is 
$$\max_p  V(x,U_i(x),p) ,$$ 
which it solves by choosing the vector $p=[p_1, \ldots, p_N]$,
i.e. providing  incentive factors  to the players. Note that the designer objective $V$ (possibly) depends on player utilities $U=[U_1,\ldots,U_N]$, yet the designer does not have direct knowledge on them.
Furthermore, the risk manager may have only a limited budget $B$  to achieve its goal that leads to the additional constraint 
$$\sum_{i=1}^N p_i x_i \leq  B .$$

\textbf{Mechanism design}, as a field of game theory, studies designing the rules and structure of games such that their outcome achieve certain objectives~\cite{maskin1,lazarSemret1998,johari1,alpcan-infocom10}. Two criteria a mechanism has to satisfy has already been
described above. The player objective of minimizing own cost can also be called as \textit{preference-compatibility}. Likewise,
the designer objectives of maximizing $V$ or achieving a global goal can be interpreted as an \textit{efficiency} criterion. The third criterion arises from the
fact that the interaction between the designer and  players of the game (Figure~\ref{fig:mechdesign1}) may motivate the players to misrepresent their utilities to the designer. They can benefit from misrepresenting their utilities (exaggerating or diminishing the actual benefits of their investments) to receive higher  incentives. Therefore, mechanism design has a third objective called interchangeably \textit{strategy-proofness}, \textit{truth dominance}, or \textit{incentive-compatibility} in addition to the objectives of efficiency and preference-compatibility. All these three criteria
are summarized in the following table:
\begin{table}[htp]
\begin{center}
\caption{Three Criteria of Mechanism Design}
\begin{tabular}[t]{|l|l|}
\hline 
 \textit{Criterion} & \textit{Formulation in the Model} \\
\hline \hline 
Efficiency   & Designer objective \\ \hline 
Preference-  & Players minimizing own costs \\ 
compatibility & (NE as operating point) \\  \hline 
Strategy-Proofness & No player gains from cheating \\
\hline
\end{tabular}
\end{center} \label{tbl:mechdesign}
\end{table}

\subsection{Assumptions}

Taking into account the breadth of the field mechanism design, it is useful to clarify the underlying assumptions of the model studied
in this section. The \textbf{environment} where the players and designer interact is characterized by the following properties:
\begin{itemize}
 \item The players and designer operate with limited resources, e.g. under budget and manpower constraints.
 \item The organizational structure imposes restrictions on available information to players and communication between them. 
 \item The designer has no information on the preferences of individual players, but observes their actions and final costs.
\end{itemize}

The players share and compete for limited resources in the given environment under its information and communication constraints. 
The following assumptions are made on the \textbf{designer and players}:
\begin{itemize}
 \item The designer is honest, i.e. does not try to deceive players.
 \item Each player acts alone and rationally according to own self interests. 
 \item The players may try to deceive the designer by hiding or misrepresenting their own preferences.
 \item All players follow the rules of the mechanism imposed by the designer.
\end{itemize}

Implications of these assumptions and limitations of the presented model will be further discussed in Section~\ref{sec:discussion}.

\section{Incentive Mechanism Design} \label{sec:incentivemech}

This section presents two specific incentive mechanisms for risk management based on the model of the previous section.
In the first mechanism, $\mc M_1$, the risk manager (designer) aims to maximize the aggregate benefit from security investments
of units, which is the sum of player utilities. This objective is sometimes also called as ``social welfare maximization''.
The second mechanism, $\mc M_2$ represents a scenario in which the risk manager aims to align efforts of all units for 
deployment and operation of an organization-wide risk management solution. Both mechanisms (their iterative variants) satisfy the
criteria in Table~\ref{tbl:mechdesign} under specific conditions. The interaction between the designer
and players is visualized in Figure~\ref{fig:subsidymech1}.

\begin{figure}[htp]
  \centering
  \includegraphics[width=\columnwidth]{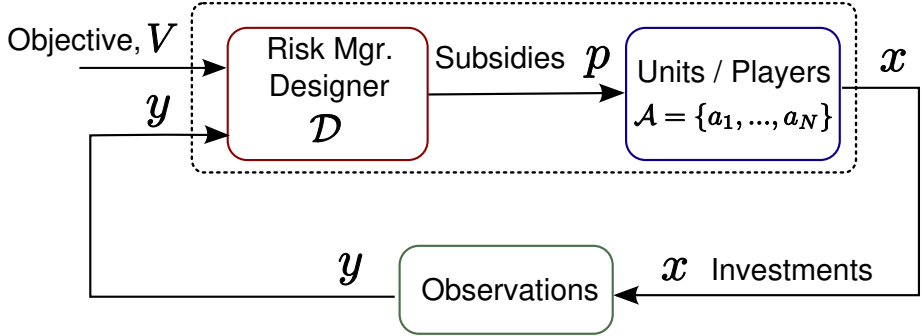}
  \caption{Interaction between risk manager (designer) and organizational units (players) as part of incentive mechanism design.}
\label{fig:subsidymech1}
\end{figure}

\subsection{Welfare maximizing mechanism}

The optimization problem $\min_{x_i} J_i (x)$ of player $i$ is a convex one and admits the unique solution
$$ x_i^* = \left( \dfrac{\partial U_i(x)}{\partial x_i}\right)^{-1} (\beta_i - p_i) ,$$
under the strict concavity and continuous differentiability assumptions on $U_i$ \cite{bertsekas2}.
Any such solution $x^*$ that solves all player optimization problems is by definition \textbf{preference-compatible}.

It is important to note that, if there was no incentive term, $p_i x_i$,  in player cost, each unit would
act according to self interest only resulting in a suboptimal result for the entire organization;
a situation sometime termed as  \textit{tragedy of commons}. The designer can prevent this by
providing a carefully selected incentive  scheme \cite{cdc09lacra,gamenetsne}.

The risk manager $\D$ objective in mechanism $\mc M_1$ is to maximize sum of player utilities, $\sum_i U_i(x_i)$. Considering
that under the assumptions of Section~\ref{sec:model} the risk manager does \textit{not} know these utilities makes
this goal paradoxical at first glance. However, the risk manager can actually achieve it in a carefully designed mechanism
where it deduces the needed parameters for the solution from the observed actions of players. 

Formally, the designer solves the constrained optimization problem
\begin{equation} \label{e:designerobj1}
 \max_x V(x) \Leftrightarrow \max_x \sum_i U_i (x) \text{ such that } \sum_i p_i x_i \leq B.
\end{equation}
The optimal solution to this constrained problem by definition satisfies the \textbf{efficiency criterion}.
The associated Lagrangian function is then
$$ L(x)=\sum_i U_i (x) + \l \left( B- \sum_i p_i  x_i \right)  ,$$
where $\l\geq 0$ is a scalar Lagrange multiplier \cite{bertsekas2}.
Under the concavity assumptions on $U_i$, this leads to
\begin{equation} \label{e:global1}
\dfrac{\partial L}{\partial x_i}=0  \Rightarrow \dfrac{1}{p_i}\sum_{j=1}^N \dfrac{\partial U_j(x)}{\partial x_i}= \l, \; \forall i \in \mc A,
\end{equation}
and the associated budget constraint\footnote{An underlying assumption here is that the risk manager (designer) utilizes
all of its budget, i.e. the constraint is active.} is
\begin{equation} \label{e:constraint1}
\dfrac{\partial L}{\partial \l}=0 \Rightarrow \sum_i p_i x_i=B.
\end{equation}

Meeting both the preference-compatibility and efficiency 
criteria requires alignment of player and designer optimization problems. This alignment
can be achieved by choosing the Lagrange multiplier $\l$ and player  incentive factors $p$
in such a way that
\begin{equation} \label{e:align1}
\dfrac{\b_i - p_i}{p_i} + \dfrac{1}{p_i}\sum_{j\neq i}\dfrac{\partial U_j(x)}{\partial x_i} = \l, \; \forall i \in \mc A,
\end{equation}
and
\begin{equation} \label{e:align2}
\sum_i p_i \left( \dfrac{\partial U_i(x)}{\partial x_i} \right)^{-1}( \b_i - p_i) =B. 
\end{equation}
Any solution to the set of $N+1$ nonlinear equations (\ref{e:align1})-(\ref{e:align2}) is by
definition a Nash equilibrium as it lies at the intersection of the player best responses. These results are summarized in the following proposition.

\begin{prop} \label{prop:m1}
 Any solution of the mechanism $\mc M_1$ described above obtained from (\ref{e:align1})-(\ref{e:align2}) is both player preference-compatible (based on the strategic game $\mc G$, given in Definition~\ref{def:game}) and efficient, i.e. maximizes $\sum_i U_i (x)$.
\end{prop}

If the designer $\D$ wants to compute the  incentive factors $p$ directly by solving (\ref{e:align1})-(\ref{e:align2}), it needs to ask each individual player $i$ for its utility, more specifically 
$\partial  U_i(x) / \partial x_j$ $\forall j \in \mc A$. 
However, the players have now a motivation to misrepresent their utilities  to the designer in order to gain a larger
share of resources or  incentive factors. To see this, consider a cheating player $i$ reporting $\tilde U_i$ to
the designer instead of their true values. If the designer believes the player and solves  (\ref{e:align1})-(\ref{e:align2})
using these, then the resulting incentive factor $\tilde p$ will naturally be different from what it should have been, $p$. 
A selfish or malicious player can thus manipulate such a scheme, which by definition is not strategy-proof.
Note that, the risk manager has access to costs $\b_i x_i$ and actions $x_i$ of individual players, which
can be, for example, part of an organizational reporting process.

One way to address the issue of strategy-proofness is to devise additional schemes to  detect potential player misbehavior 
(for which players already have a motivation). This, however, brings an additional layer of overhead to the
overall system both in terms of communication and computing requirements.

Alternatively, one can design an \textbf{iterative mechanism} that is based
on observation of player actions $x$ instead of asking for their word (utilities). This approach
is the basis of the iterative schemes that will be presented in Section~\ref{sec:iterative}.

\subsection{Mechanism with global objective}

The second mechanism, $\mc M_2$ differs from the social welfare maximizing one $\mc M_1$
discussed in the previous subsection. In this case, the designer
has an organization-wide or ``global'' objective represented by the strictly concave and nondecreasing function $F(x)$ which
does not directly depend on player utilities. This organization-wide objective could be, 
for example, deployment and operation of an organization-wide risk management
solution that naturally requires cooperation from all units and an alignment of efforts.

In mechanism $\mc M_2$, the  risk manager formally solves the constrained
optimization problem
\begin{equation} \label{e:designerobj2}
 \max_x F(x) \text{ such that } \sum_i p_i x_i \leq B.
\end{equation}
The associated Lagrangian function is then
$$ L(x)=F(x) + \l \left( B- \sum_i p_i  x_i \right)  ,$$
where $\l>0$ is a scalar Lagrange multiplier. Note that the constraint is always active in this case due to the definition of $F(x)$. Under the concavity assumptions on $F(x)$, this leads to
\begin{equation} \label{e:global2}
\dfrac{\partial L}{\partial x_i}=0  \Rightarrow \dfrac{1}{p_i}\dfrac{\partial F(x)}{\partial x_i}= \l, \; \forall i \in \mc A,
\end{equation}
and the associated budget constraint is
\begin{equation} \label{e:constraint2}
\dfrac{\partial L}{\partial \l}=0 \Rightarrow \sum_i p_i x_i=B.
\end{equation}

Combining this with the player optimization problems to ensure efficiency and preference-compatibility
as in the previous subsection leads to
\begin{equation} \label{e:align1a}
\dfrac{1}{p_i} \dfrac{\partial F(x)}{\partial x_i}= \l, \; \forall i \in \mc A,
\end{equation}
and
\begin{equation} \label{e:align2a}
\sum_i p_i \left( \dfrac{\partial U_i(x)}{\partial x_i} \right)^{-1}( \b_i - p_i) =B,
\end{equation}
which are direct counterparts of (\ref{e:align1})-(\ref{e:align2}). As before, any solution constitutes a Nash equilibrium as it lies at the intersection of the player best responses.

\begin{prop} \label{prop:m2}
 Any solution of the mechanism $\mc M_2$ described above obtained from (\ref{e:align1a})-(\ref{e:align2a}) is both player preference-compatible (based on the strategic game $\mc G$, given in Definition~\ref{def:game}) and efficient, i.e. maximizes $F(x)$.
\end{prop}

In mechanism $\mc M_2$, the risk manager has to evaluate the term $\partial F(x) / \partial x_i$
for each unit $i$, in addition to asking them for their utilities and cost factors. This term can be interpreted
as the rate of contribution of each unit to the organization-wide objective. Since the risk manager sets this
objective, it can be computed or estimated with reasonable accuracy. However, as before
the solution of (\ref{e:align1a})-(\ref{e:align2a}) also depends on individual unit utilities and cost factors.
Therefore, mechanism $\mc M_2$ --similar to $\mc M_1$ -- requires deployment of iterative methods in order to meet the
criterion of strategy-proofness.

\subsection{Interdependent Utilities and Linear Influence Model} \label{sec:coupledutil}

In the presented model and analysis, utilities
of individual players (units) may depend not only on their own actions but also on those of others,
e.g. $U_i(x)=U_i ([x_1,\ldots,x_N])$. 
In other words, a unit benefits not only from own risk investments but also from efforts
of other related units. Such utility functions are called interdependent or nonseparable in
contrast to separable player utilities, $U_i(x_i)$, that depend only own actions.
If the player utilities are separable, then the player decisions are almost completely decoupled from each 
other except from external resource constraints (such as the   incentives they receive from the designer).
This simplifies development of decentralized schemes significantly.

One possible way of modeling  interdependencies in player utilities is the linear influence model, 
which captures how actions (investments) of players (units) affect others.
As a first-order approximation these effects are modeled as \textit{linear} resulting in an 
\textit{influence matrix} defined as
\begin{equation} \label{e:influencematrix}
W := 
\begin{cases}
1 & \text{, if } i =j , \\
w_{ij} & \text{, otherwise.}
\end{cases}
\end{equation}
where $0 \leq w_{ij} \leq 1$ denotes the  non-negative effect of unit $j$ ('s investment) on unit $i$. Notice
that this effect may well be zero. 

Define now the vector of \textit{effective investments}
$x^e=[x_1^e, \ldots, x_N^e]$, where the effective investment of unit $i$ is 
$$x^e_i :=\sum_j W_{ij} x_j = (W x)_i ,$$
and $(\cdot)_i$ denotes the $i^{th}$ element of a vector.

Naturally, it is possible to develop more complex nonlinear models to capture interdependencies between
units and their actions. However, given the limitations on information collection and accuracy,
the linear first order approximation described provides a good starting point. Therefore,
we will use linear influence model in the case of interdependent (non-separable) utilities 
for the rest of the paper.

Note that under the linear influence model, the nonseparable utility, $U_i(x)$, of player $i$ is 
given by 
$$U_i(x^e_i)=U_i \left( (W x)_i \right).$$

\section{Iterative Incentive Mechanisms} \label{sec:iterative}

Mechanisms $\mc M_1$ and $\mc M_2$ as defined in the previous section are shown to be efficient and
preference-compatible (See Propositions~\ref{prop:m1} and \ref{prop:m2}) but not strategy-proof. This section presents two iterative variants of these mechanisms that satisfy all three criterion and can be implemented under information limitations.

\subsection{Iterative mechanism with global objective}

In the iterative mechanism with global objective, $\mc{IM}_2$, both the risk manager
and units adopt an iterative scheme to facilitate information exchange
that does not allow cheating, hence resulting in a strategy-proof mechanism.
Specifically, the risk manager updates the Lagrangian multiplier $\l$ in (\ref{e:global2}) 
gradually according to
\begin{equation} \label{e:iterative1a}
 \l (n+1) = \l (n) + \kappa_d \left[ \sum_i p_i(n)  x_i (n) - B \right]^+  ,
\end{equation}
and computes the individual player  incentive factors
\begin{equation} \label{e:iterative2a}
 p_i (n)= \dfrac{1}{\l(n)} \dfrac{\partial F(x(n))}{\partial x_i} .
\end{equation}
Here, $n=1,\ldots$ denotes the iteration number or time-step. The  units (players)
in return react to given  incentive factors by updating their investment decisions
in order to minimize their own costs such that
\begin{equation} \label{e:iterative3a}
\begin{array}{lll}
 x_i (n+1) & = & \phi x_i (n) \\
   &+ & (1-\phi)\left( \dfrac{\partial U_i(x(n))}{\partial x_i}\right)^{-1} (\beta_i - p_i(n))   \;\; \forall i,
\end{array}
\end{equation}
where $0< \phi <1$ is a relaxation constant used by the players to 
prevent excessive fluctuations. Alternatively, this behavior can be justified with caution or 
inertia of the organizational units.

The equilibrium solution(s) of (\ref{e:iterative1a})-(\ref{e:iterative3a})
clearly coincides with that of (\ref{e:align1a})-(\ref{e:align2a}). Hence, the iterative mechanism $\mc{IM}_2$,
assuming that it converges, solves the same problem as mechanism $\mc M_2$. Furthermore, it is strategy-proof since
at each update step, the players make decisions according to their own self interests and do not have the
opportunity of manipulating the system.  To see this, assume otherwise and let player $i$ ``misrepresent''
its actions $\tilde x_i = x_i + \delta$ for some $\delta \in \Real$. Then, the player's
instantaneous cost is $J_i (\tilde x_i) > J_i (x_i)$ at each step of the iteration. Hence, the
players have no incentive to ``cheat''. These results are summarized in the following theorem which extends Proposition~\ref{prop:m2}:

\begin{thm} \label{thm:m2}
Any solution of the iterative mechanism with global objective, $\mc{IM}_2$ described above and in Algorithm~\ref{alg:iterative1}
is player preference-compatible, efficient, and strategy-proof.
\end{thm}

Information flow and limitations play a crucial role in implementation of the iterative mechanism $\mc{IM}_2$. 
In practice, the risk manager is assumed to observe the actions of units which they have to reveal
in order to receive  incentives. Based on this information and the total budget, the risk manager can easily implement (\ref{e:iterative1a}). Then, it only needs to estimate the individual marginal contributions of units to the overall objective, $\partial F(x(n)) /\partial x_i $ at a given moment in order to decide on actual  incentive factors in (\ref{e:iterative2a}). 

Likewise, given own cost estimates $\b_i$ and incentive factor $p_i$, 
each unit (player) only has to determine the marginal benefit from its own actions, $\partial U_i(x(n)) / \partial x_i$ in order to implement (\ref{e:iterative3a}). If the unit has a separable utility, then this is simply equivalent to $\partial U_i(x_i(n)) / \partial x_i$.
In the interdependent utility case, under the linear influence model this quantity turns out to be
the marginal benefit from the effective action, 
$$ 
\dfrac{\partial U_i(x(n))}{\partial x_i}= \dfrac{\partial U_i( x_i^e(n))}{\partial x_i^e}\dfrac{\sum_j W_{ij} x_j}{x_i}=\dfrac{\partial U_i( x_i^e(n))}{\partial x_i^e},   
$$ 
as a result of $W_{ii}=1$ and the definitions of respective quantities. Algorithm~\ref{alg:iterative1}
summarizes the steps of the iterative mechanism with global objective, $\mc{IM}_2$.
\begin{algorithm}[!ht]
  \SetAlgoLined
  \KwIn{\textit{Designer}: budget $B$ and global objective $F(x)$}
  \KwIn{\textit{Players}: cost factor $\b_i$ and utilities $U_i, \forall i$}
  \KwResult{Player investments $x$ and incentive factors $p$}
  
  Initial investments $x_0$ and  incentive factors $p_0$ \;
  \Repeat{end of iteration (negotiation)}{
   \Begin(\textit{Designer:}){
    Observe player investments $x$ \;
    Update $\l$ according to (\ref{e:iterative1a}) \;
    Estimate marginal contributions of players to global objective, $\partial F(x) /\partial x_i $ \;
    \ForEach{player $i$}{
    Compute  incentive factor $p_i$ from (\ref{e:iterative2a}) \;
    }
   }
   \Begin(\textit{Players:}){
    \ForEach{player $i$}{
      Estimate marginal utility $\partial U_i(x)/ \partial x_i$ \;
      Compute investment $x_i$ from (\ref{e:iterative3a}) \;
    }  
   }
  }
  \caption{Iterative mechanism  $\mc{IM}_2$} \label{alg:iterative1}
\end{algorithm}

\subsubsection*{Convergence Analysis of $\mc{IM}_2$}

A basic stability analysis is provided for a continuous-time approximation of
the iterative mechanism with global objective, $\mc{IM}_2$. For tractability, let the
player utilities be of the form $U_i=\a_i \log(x_i)$. Further define the global objective
function of the risk manager as $F(x):=\sum_i \gamma_i x_i$, for some $\gamma_i >0 \; \forall i$.

Substituting $p_i$ with $\gamma_i / \lambda$, the continuous-time counterpart of (\ref{e:iterative1a})-(\ref{e:iterative3a}) is
\begin{eqnarray} \label{e:contiterative1}
 \dot \lambda = \dfrac{d \l}{dt}= \kappa_{\lambda} \dfrac{1}{\lambda}\left( \sum_i \gamma_i x_i (n) - B \right) \\
 \dot x_i =-\kappa_i \dfrac{\partial J_i}{\partial x_i }=\kappa_i 
 \left(\dfrac{\a_i}{x_i} + \dfrac{\gamma_i}{\lambda} - \b_i  \right) \;\; , \forall i \in \mc A. \nonumber
\end{eqnarray}
where $t$ denotes time and $\kappa_{\lambda},\; \kappa_i>0$ are step-size constants. As in the discrete-time
version, the players adopt here a gradient best response algorithm. Define the Lyapunov function
$$ V_L:= \frac{1}{2} \left( \dfrac{ \sum_i\gamma_i x_i  - B}{\lambda}  \right)^2 +
 \frac{1}{2} \sum_i  \left(\dfrac{\a_i}{x_i} + \dfrac{\gamma_i}{\lambda} - \b_i   \right)^2 ,$$
which is nonnegative except at the solution(s) of (\ref{e:iterative1a})-(\ref{e:iterative3a}), i.e. $V_L (x^*,\l^*)=0$. 

Taking the derivative of $V_L$ with respect to time yields
$$ \dot V_L (x,\l) = -2 \dfrac{\sum_i\gamma_i x_i}{\lambda^3} \left( \dfrac{ \sum_i\gamma_i x_i  - B}{\lambda}  \right)^2 
- \sum_i \dfrac{\a_i}{x_i^2} \left(\dfrac{\a_i}{x_i} + \dfrac{\gamma_i}{\lambda} - \b_i   \right)^2.
$$
Consider the region where $F(x)=\sum_i \gamma_i x_i >0$. Then, there exists an $\epsilon>0$ such that
$$ \dot V_L (x,\l) \leq \epsilon V_L <0 , \;\; \forall (x,\l) \neq (x^*,\l^*),$$
i.e. for any point of the trajectory $(x,\l)$ not equal to a solution of (\ref{e:iterative1a}) and (\ref{e:iterative3a}). Thus, the continuous-time algorithm is exponentially stable \cite{khalilbook} on the set
$\mc{\bar X}:=\{ x \in \mc X : F(x) >0 \}$.
This result, which is summarized in the next theorem, is a strong indicator of fast convergence \cite{bertsekas3} of the discrete-time iterative pricing mechanism (\ref{e:iterative2a})-(\ref{e:iterative2b}). 

\begin{prop} \label{thm:converge1}
The continuous-time approximation of the iterative mechanism $\mc{IM}_2$, given by (\ref{e:contiterative1}) exponentially converges to a solution of (\ref{e:iterative1a})-(\ref{e:iterative3a})  on the set
 $\mc{\bar X}=\{ x \in \mc X : F(x) >0 \}$.
\end{prop}

The \textit{exponential convergence} result above indicates a very fast convergence rate. To see this, let $x(0)$ be the initial player investments and $x^*$ denote  a solution of (\ref{e:iterative1a})-(\ref{e:iterative3a}). Then, for the player investments $x(t)$ under continuous-time approximation of the iterative mechanism $\mc{IM}_2$ the following holds:
$$ \norm{x(t) - x^*} \leq \alpha \norm{x(0) - x^*} e^{-\beta t},$$
for $t \geq 0$ and some $\alpha, \beta >0$. In other words, the investment levels approach their equilibrium values exponentially fast.

\subsection{Iterative welfare maximizing mechanism}

The iterative welfare maximizing mechanism $\mc{IM}_1$ extends mechanism $\mc M_1$. Same
as the previous mechanism, the risk manager updates the Lagrangian multiplier $\l$ according to (\ref{e:iterative1a})
and the unit updates are given by (\ref{e:iterative3a}).

However, the computation of individual player  incentive factors is more involved due to the dependence
of the objective (welfare maximization) on individual player utilities
\begin{equation} \label{e:iterative2b}
 p_i (n)= \dfrac{1}{\l(n)} \sum_j \dfrac{\partial U_j(x(n))}{\partial x_i} ,
\end{equation}
which follows from (\ref{e:align1}).
At first glance, it seems that the designer has to ask players again for their marginal utility which defeats
the purpose of the iterative approach, namely ensuring strategy-proofness. Fortunately, the designer
can circumvent this issue by utilizing side information, in this case player cost factors $\b$, within the
linear influence model.

It directly follows from the linear influence model that
$$ \dfrac{\partial U_j}{\partial x_i}=\dfrac{\partial U_j}{\partial x_j^e} \dfrac{\partial x_j^e}{\partial x_i}
=\dfrac{\partial U_j}{\partial x_j^e} W_{ji}= \dfrac{\partial U_j}{\partial x_j} W_{ji}.$$ 
The actions of any player $i$ chosen according to a (relaxed) best response (\ref{e:iterative3a}), and
observed by the designer yields the information
$$ \dfrac{\partial U_i(x)}{\partial x_i}=\beta_i - p_i $$
to the designer. Hence, the substitution 
$$ \sum_j \dfrac{\partial U_j(x(n))}{\partial x_i}=\sum_j ( \beta_i - p_i) W_{ji}, \; \forall i,j \, , $$
can be used in (\ref{e:iterative2b}) to obtain
$$ \l p = W^T (\b -p).$$
Thus, the designer implements 
\begin{equation} \label{e:iterative2c}
p= (W^T + \l I)^{-1} W^T \b 
\end{equation}
together with (\ref{e:iterative1a}) to determine player  incentive factors. Here, $(\cdot)^T$ denotes the transpose operator and $I$ the identity matrix. These results are summarized in the following theorem which extends Proposition~\ref{prop:m1}:

\begin{thm} \label{thm:m1}
Any solution of the iterative mechanism with global objective, $\mc{IM}_1$ described above and in Algorithm~\ref{alg:iterative2}
is player preference-compatible, efficient, and strategy-proof.
\end{thm}

The information structure in mechanism $\mc{IM}_1$ is similar to that of $\mc{IM}_2$ with the
following differences. In $\mc{IM}_1$, the risk manager has to estimate the linear dependencies in the system
represented by the matrix $W$ and observe cost factors $\b$ of units in addition to their investments.
These information requirements are due to the complex nature of the welfare maximization objective,
which necessitates additional (indirect) communication between the risk manager and
units in practice. Algorithm~\ref{alg:iterative2} summarizes the steps of the welfare maximizing mechanism
 $\mc{IM}_1$.

\begin{algorithm}[hpb]
  \SetAlgoLined
  \KwIn{\textit{Designer}: budget $B$ and objective $\sum_i U_i$}
  \KwIn{\textit{Players}: cost factor $\b_i$ and utilities $U_i, \forall i$}
  \KwResult{Player investments $x$ and  incentive factors $p$}
  
  Initial investments $x_0$ and  incentive factors $p_0$ \;
  \Repeat{end of iteration (negotiation)}{
   \Begin(\textit{Designer:}){
    Observe player actions $x$ and cost factors $\b$ \;
    Estimate the linear influence matrix $W$ \;
    Update $\l$ according to (\ref{e:iterative1a}) \;
    Compute  incentive factors $p$ from (\ref{e:iterative2c}) \;
    }
   \Begin(\textit{Players:}){
    \ForEach{player $i$}{
      Estimate marginal utility $\partial U_i(x)/ \partial x_i$ \;
      Compute investment $x_i$ from (\ref{e:iterative3a}) \;
    }  
   }
  }
  \caption{Iterative mechanism  $\mc{IM}_1$} \label{alg:iterative2}
\end{algorithm}

\subsubsection*{Convergence Analysis of $\mc{IM}_1$}

A basic stability analysis is provided for a continuous-time approximation of
the iterative mechanism  $\mc{IM}_1$ similar to the one of the $\mc{IM}_2$ in the previous subsection. 
For tractability, let the player utilities be of the form $U_i=\a_i \log(x_i)$ as before.

Substituting $p_i$ with
$$p_i= \dfrac{\b_i }{1+\l}, $$
which follows from (\ref{e:iterative2c}) and $W=I$,
the continuous-time counterpart of (\ref{e:iterative1a}) and (\ref{e:iterative3a}) is
\begin{eqnarray} \label{e:contiterative2}
 \dot \lambda = \dfrac{d \l}{dt}= \kappa_{\lambda} \dfrac{1}{1+ \lambda}\left( \sum_i \b_i x_i (n) - B \right) \\
 \dot x_i =-\kappa_i \dfrac{\partial J_i}{\partial x_i }=\kappa_i 
 \left(\dfrac{\a_i}{x_i} + \dfrac{\b_i}{1+ \lambda} - \b_i  \right) \;\; , \forall i \in \mc A. \nonumber
\end{eqnarray}
where $t$ denotes time and $\kappa_{\lambda},\; \kappa_i>0$ are step-size constants. As in the discrete-time
version, the players adopt here a gradient best response algorithm. Define the Lyapunov function
$$ \bar V_L:= \frac{1}{2} \left( \dfrac{ \sum_i\b_i x_i  - B}{1+ \lambda}  \right)^2 +
 \frac{1}{2} \sum_i  \left(\dfrac{\a_i}{x_i} + \dfrac{\b_i}{1+ \lambda} - \b_i   \right)^2 ,$$
which is nonnegative except at the solution(s) of (\ref{e:iterative1a}) and (\ref{e:iterative3a}), i.e. $\bar V_L (x^*,\l^*)=0$. 

Taking the derivative of $\bar V_L$ with respect to time yields
$$ \dot{\bar{V}}_L (x,\l) = -2 \dfrac{\sum_i\b_i x_i}{(1+\lambda)^3} \left( \dfrac{ \sum_i\b_i x_i  - B}{1+\lambda}  \right)^2 
- \sum_i \dfrac{\a_i}{x_i^2} \left(\dfrac{\a_i}{x_i} + \dfrac{\b_i}{1+\lambda} - \b_i   \right)^2.
$$
Consider the region where $\sum_i \b_i x_i >0$. Then, there exists an $\epsilon>0$ such that
$$ \dot{\bar{V}}_L (x,\l) \leq \epsilon \bar V_L <0 ,\;\; \forall (x,\l) \neq (x^*,\l^*), $$
i.e. for any point of the trajectory $(x,\l)$ not equal to a solution of (\ref{e:iterative1a}) and (\ref{e:iterative3a}). 
Thus, the continuous-time algorithm is exponentially stable \cite{khalilbook} on the set
$\mc{\tilde X}:=\{ x \in \mc X : \sum_i \b_i x_i >0 \}$.
This result, which is summarized in the next theorem, is a strong indicator of fast convergence \cite{bertsekas3} of the discrete-time iterative pricing mechanism $\mc{IM}_1$. 

\begin{prop} \label{thm:converge2}
The continuous-time approximation of the iterative mechanism $\mc{IM}_1$, given by (\ref{e:contiterative2}) exponentially converges to a solution on the set  $\mc{\tilde X}:=\{ x \in \mc X : \sum_i \b_i x_i >0 \}$. 
\end{prop}

\section{Use Case Scenario and Numerical Analysis} \label{sec:numerical}

In order to illustrate the incentive mechanism framework for risk management, a use case
scenario is described next. Since most organizations do not openly publish their actual risk management
structure or numbers, this scenario is naturally hypothetical and the numbers in the subsequent
numerical analysis do not necessarily coincide with real world counterparts. 

\subsection{Example Use Case Scenario}

In this subsection, a possible use cases scenario is described for a large-scale enterprise with
multiple autonomous business units, denoted by set $\mc A$, who collaborate and share IT infrastructure in order 
to provide various services and products. In addition to the business units, the enterprise headquarters has a special security  risk management division, which will be simply referred to as ``risk manager'' here. The task of the risk manager, $\D$,
is successful deployment and operation of security and IT risk management projects that
entail enterprise-wide computer-assisted information collection (observation), risk assessment (decision making), 
and mitigation (control). 

The results and algorithms described in this paper can be utilized to develop a manual risk management strategy 
as well as a technical system to handle a large number of business units and multiple concurrent risk 
management projects. For simplicity and as a special case of the latter, this scenario focuses on the former.

Let the risk manager start a project to improve robustness of the IT systems involved in a product against
security threats. The success of the project naturally depends on collaboration of the $6$ specific business
units involved at various stages of the product in question. However, not each unit plays an equal role
in creation of the product, and hence, their risk exposure is different. Therefore, those units with a 
more significant role have to make a larger investment to the project and their IT systems. 

During the project, the divisions have to provide accurate information on their business and networked systems.
At the operational phase, each division allocates manpower and resources for the proper operation of 
the system. Hence, participation in this risk management project is associated with a certain cost to each unit in terms of 
investments and manpower. Although each unit sees a certain amount of value in the new risk management system,
if they are left alone to themselves, their contributions may not be sufficient for the successful realization of the
risk management system. Thus, the risk manager uses parts of its budget for subsidizing individual unit investments,
if necessary in the form of manpower and expertise.

Let $x=[x_1, x_2, \ldots, x_6]$ denote the investments (project contributions) of business units. Their contribution 
to the project is evaluated using the multi-variable objective function $F(x)$, which describes the goal of the
entire project. The individual marginal contribution of a business unit $i$ (one of six) to the project
goal at a given (project) state is given by the derivative, $\partial F(x) / \partial x_i$. It is important to note
that risk manager may not know the exact form of $F(x)$ before hand, and has to estimate $\partial F(x) / \partial x_i$
for each business unit $i$ at a given state.

The goal of the risk manager is to ensure the success of the project, which may be captured by making the objective
function achieve a certain minimum threshold value, i.e. $F(x)>V_{threshold}$. The subsidies given to the units
(monetarily or in the form of assistance) are determined in proportion to their current investments. For example,
the business unit $i$ receives $p_i x_i$. These subsidies have to be of course within the allocated budget, i.e.
$\sum_{i=1}^6 p_i x_i \leq B$. Note that, the budget in question is periodic, e.g. $B$ units per month or year.

The interaction between the risk manager and individual units is designed according to Algorithm~\ref{alg:iterative1}
based on the $\mc{IM}_2$. The actual time-scale of the iteration depends on the specific requirements of the
enterprise. For example, the risk managers and representatives from the units may come together in weekly
or bi-weekly intervals to evaluate the progress, which gives some time to the units and manager for updating 
own evaluations on marginal benefits and contributions, respectively. We next present a numerical example
to further illustrate the scenario described.

\subsection{Numerical Analysis}

Based on the use case scenario, an example is numerically analyzed with a risk manager and $6$ units, who implement 
the iterative mechanism $\mc{IM}_2$ using Algorithm~\ref{alg:iterative1}. The budget is $B=3$,
the global objective function of the risk manager is $F(x)=\sum_{i=1}^6 \gamma_i x_i$,
where $\gamma=[0.8,\; 0.4, \; 0.5, \; 0.2,\; 0.3, \;0.1]$, the utilities of units are in the form of $U_i(x_i)=\a_i \log(x_i)$,
where $\a=[0.9,\; 0.7,\; 0.6,\; 0.8,\; 0.2,\; 0.4] $, and the unit cost factor is $\b=3$ for all six units. Each unit
starts the iteration with an initial investment of $x_i=0.5 \; \forall i $ and receives an initial incentive
factor of $p_i=0.3 \; \forall i$. The measurement units of the budget $B$ and investments $x$ are assumed
to be on the order of millions of dollars. The step-size constants are chosen as $\kappa_d=0.05$ and $\phi=0.3$.
The success of the project is decided by  whether the objective function passes minimum threshold of $2.5$, i.e.
$F(x^*)>2.5$.

The evolution of unit investment levels $x(n)$ is shown in Figure~\ref{fig:xinvest1} and the associated
incentive factors $p(n)$ in Figure~\ref{fig:incentfactor1}. The first unit, which contributes the most to the objective
receives a higher amount of aid from the risk manager than others. The algorithm converges fast, in $10-15$ steps,
for the given parameters, as indicated by the exponential convergence of its continuous-time counterpart. For a time interval of $1-2$ weeks per iteration, this corresponds to $3-6$ months
in practice. Although this convergence time may seem as a disadvantage at the first glance, in a practical 
project with highly varying parameters, such an online algorithm may even be beneficial in terms of adaptability
over time.
\begin{figure}[htp]
  \centering
  \includegraphics[width=0.7\columnwidth]{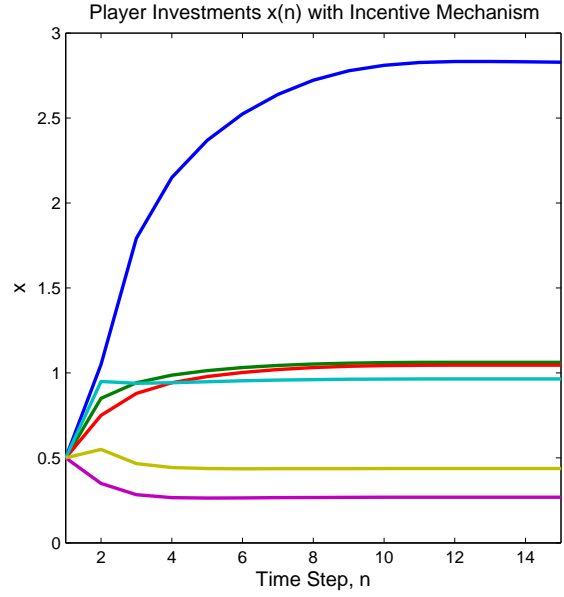}
  \caption{The evolution of unit investment levels $x(n)$ under Algorithm~\ref{alg:iterative1}. }
\label{fig:xinvest1}
\end{figure}
\begin{figure}[htp]
  \centering
  \includegraphics[width=0.7\columnwidth]{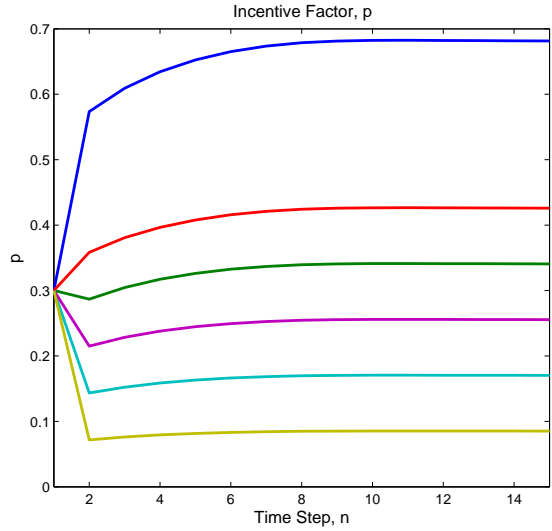}
  \caption{The evolution of incentive factors $p(n)$ under Algorithm~\ref{alg:iterative1}. }
\label{fig:incentfactor1}
\end{figure}
\begin{figure}[htbp]
  \centering
  \includegraphics[width=0.7\columnwidth]{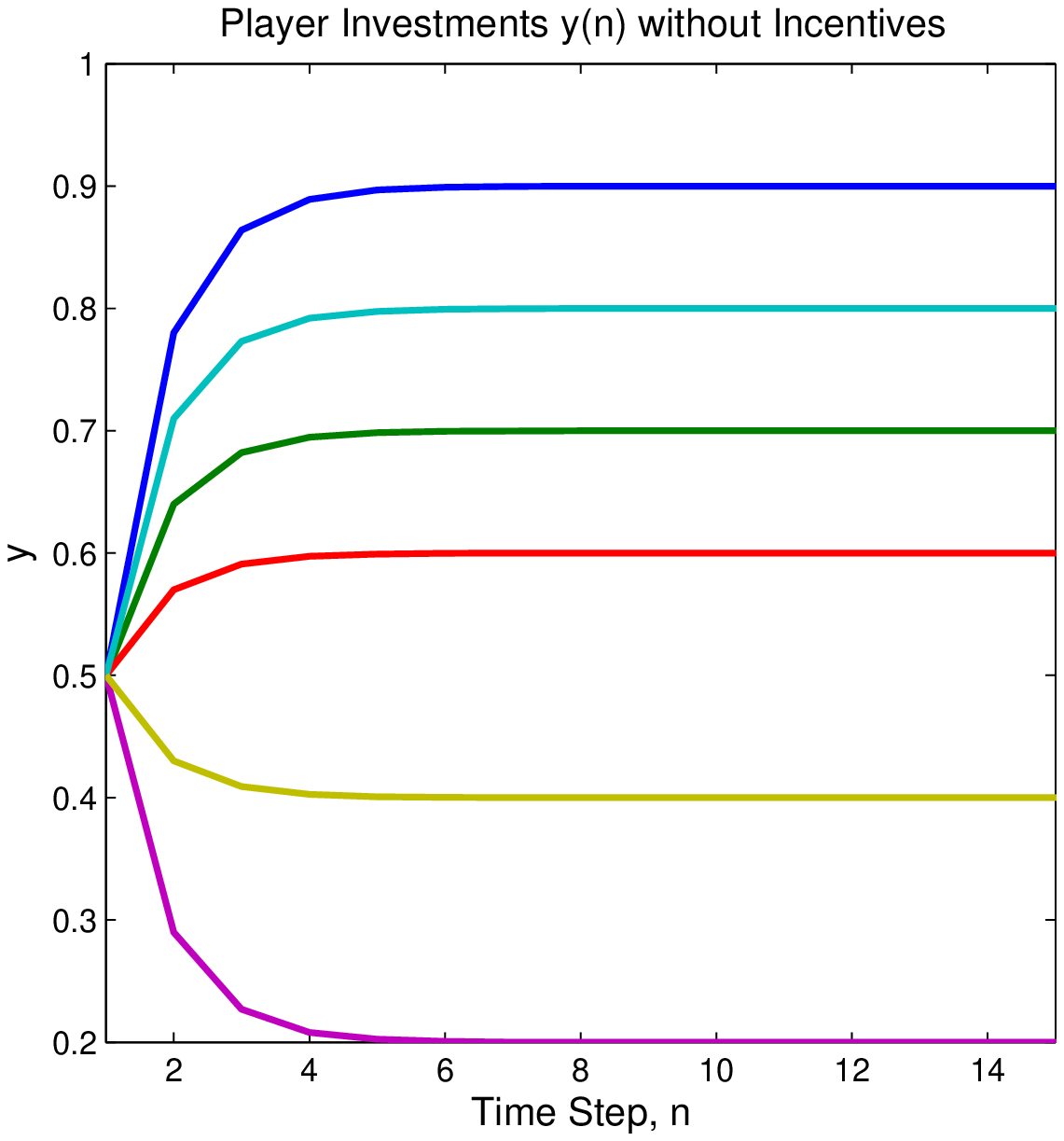}
  \caption{The evolution of unit investment levels $y(n)$ without any incentive mechanism implemented. }
\label{fig:yinvest1}
\end{figure}

In contrast, the investment results of units without any incentive mechanism in place, $y(n)$, is shown
in Figure~\ref{fig:yinvest1}. A comparison of the objective function $F(x)$ with and without an incentive
mechanism is depicted in Figure~\ref{fig:ffunction1}. Naturally, this improvement comes at an expense
of the budget $B$ spent entirely by the risk manager.
\begin{figure}[htbp]
  \centering
  \includegraphics[width=0.7\columnwidth]{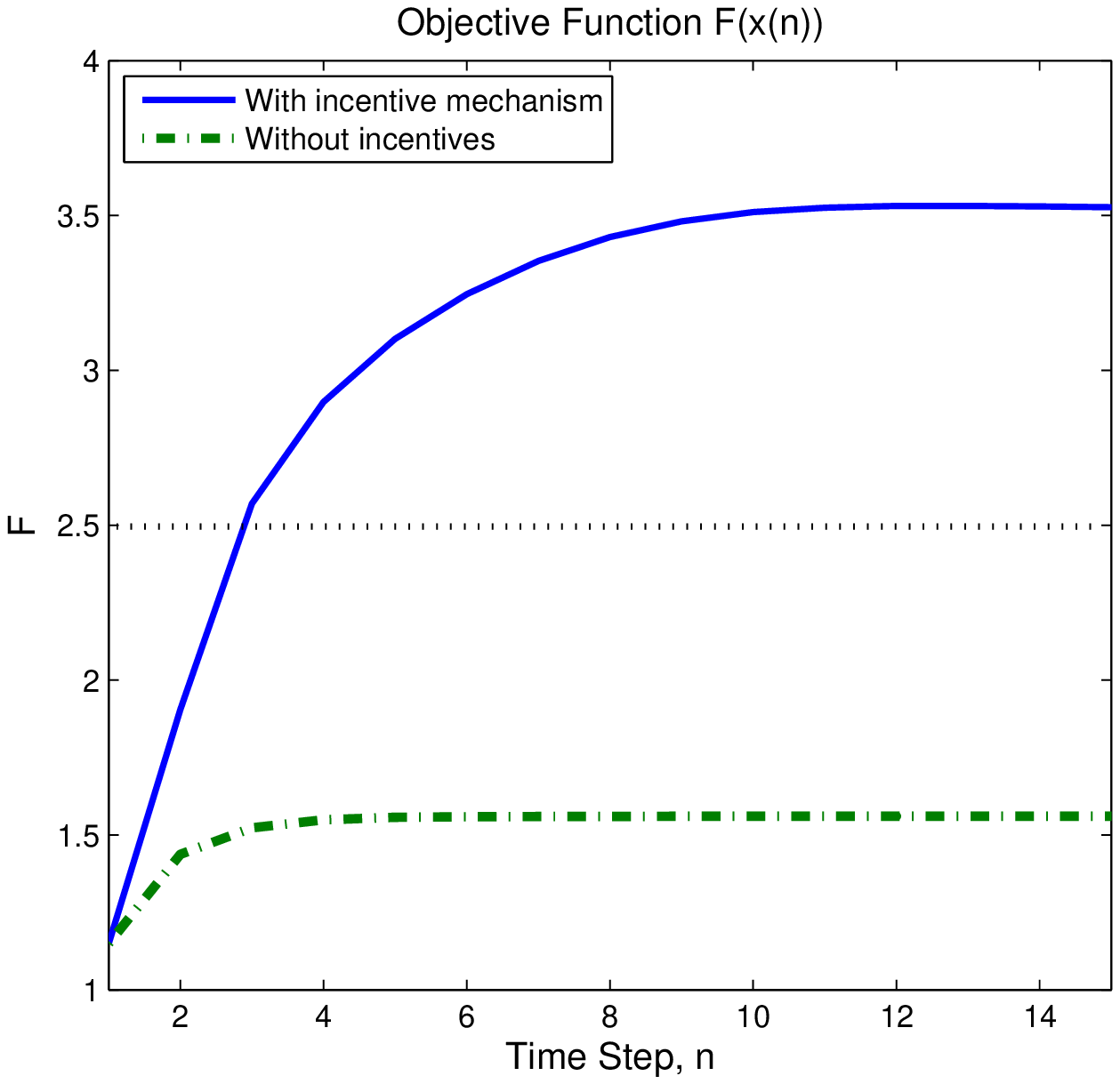}
  \caption{A comparison of the objective function $F(x(n))$ with and without an incentive
mechanism. Under the incentive mechanism it passes the success threshold of $2.5$.}
\label{fig:ffunction1}
\end{figure}

\section{Literature Review} \label{sec:discussion}

Building upon its successful applications to economics and engineering (e.g. networks), game theory has been recently utilized to model and analyze security problems \cite{alpcan-book}. Similar formalization efforts have been ongoing in the risk management area with the goal of developing analytical approaches to (security) risk analysis, management, and associated decision making \cite{riskbook1,crisis09,guikema1,icc10jeff}. Unsurprisingly, game theory enjoys an increased interest in the risk management community \cite{gtandrisk1,gtandrisk2,patecornell1,guikemachap}, as it provides a valuable and relevant mathematical framework \cite{alpcan-book,basargame,fudenberg}. Recently, a game theoretic approach has been developed for security and risk-related decision making and investments in \cite{miurako08-decision,miurako08-investment}.

Mechanism design \cite{maskin1,lazarSemret1998,johari1} is a field of game theory, where a designer imposes rules on the underlying strategic (noncooperative) game in order to achieve certain desirable objectives such as social welfare maximization or a system-wide goal. Hence, mechanism design can be viewed as a reverse engineering of games. It is especially useful in developing analytical frameworks for incentive mechanisms. Recently, there has been widespread interest in using mechanism design for modeling, analyzing and solving problems in network resource allocation problems that are decentralized in nature~\cite{johari1,hajek1,lazarSemret1998,rajiv1,wuWangLiuClancy2009,huangBerryHonig2006b}. It has also been applied to resource allocation in the context of engineering optimization~\cite{guikema2}. A basic game design approach to security investments in the risk management context has been discussed in \cite{alpcan-book}.

The presented incentive mechanism framework makes use of both mechanism design~\cite{maskin1,lazarSemret1998,johari1,alpcan-infocom10} and game theory~\cite{basargame,fudenberg}, which provide solid analytical and conceptual foundations. In contrast to many existing studies~\cite{holger-allerton,hurwicz1972,dasgupta} 
focusing on answering the question of ``which mechanisms are possible to design'', this work adopts a constructive
approach to develop a practical methodology and applies it to security risk management. Despite sharing the game-theoretic approach of earlier work \cite{miurako08-decision,miurako08-investment}, it distinguishes from these through the mechanism design framework developed on top of the game. A similar perspective has been briefly discussed in \cite[Chap. 6]{alpcan-book}, which however has not taken into account incentive-compatibility aspects. 

The article~\cite{guikema2}, which shares a similar goal as this one, discusses the problem of designing an allocation scheme that leads to truthful reporting by the engineers and allocation of the scarce resources within the VCG framework. This work distinguishes from~\cite{guikema2} in multiple ways in addition to its focus on risk management. First, the mechanisms discussed here are iterative, enable operation even under limited information, and do not require any direct revelation of preferences by the users or risk manager. Similar iterative schemes have been analyzed in depth in the networking literature, e.g. see \cite{hajek1,srikantbook,tansuphd}. Second, the sufficient conditions for convergence and operation of the iterative mechanisms here are not as restrictive as in~\cite{guikema2}. Finally, the properties of iterative interaction algorithms are analyzed rigorously from a dynamical system perspective and their rapid convergence is proven.


\section{Discussion and Conclusion} \label{sec:conclusion}

The analytical incentive mechanism design framework presented can not only be used to derive guidelines for handling incentives in risk management but also to develop computer-assisted schemes. The abstract nature of the
framework is an advantage in terms of widespread applicability to diverse situations and organization types. In order to
satisfy all three objectives of efficiency, preference-compatibility and strategy-proofness, iterative incentive mechanisms 
and related algorithms are developed which also allow implementation under information limitations. These mechanisms are very straightforward to analyze and implement numerically, which is especially useful since any practical implementation of such incentive mechanism will most probably involve some kind of computer-assistance. The risk manager has then
the option to evaluate various scenarios through simulations before actual deployment. This is illustrated with a hypothetical deployment scenario and a numerical example.

The presented inventive mechanism framework can be extended in multiple directions.
One immediate extension is multiple decision variables. For example, units may need to
distinguish between monetary investments and local resources such as manpower. Similarly,
the risk manager may utilize multiple separate incentive factors. A related but more challenging
extension is multi-criteria decision making, where preferences are not simply expressed
through scalar valued functions such as $U$ and $F$. This is an open research area also
in decision and optimization theories.

The limitations of the utility-based approach adopted here is also worth noting. The expression
of preferences through specific (continuous, differentiable) functions is obviously a simplification
to facilitate devising analytically tractable models. However, as it can be seen in Sections~\ref{sec:iterative}
and \ref{sec:numerical}, the resulting algorithms do not necessarily require the players 
estimate their whole utility beforehand. A step-by-step iterative estimation process is 
fully sufficient to establish and communicate these preferences.

An underlying assumption of the model until now has been the fixed nature of player
preferences or utility functions. Under this assumption, the risk manager can influence
unit decisions only by introducing additive incentive factors to their cost structure
as discussed. In reality however, the unit preferences are open to changes through
psychological factors. The arts of persuasion and politics may ``shift'' the utility 
curves in the model. Quantification of such factors is obviously a significant
yet open research challenge.

An approach closely related to the strategic (noncooperative) game framework discussed in
this paper, is based on coalitional (cooperative) games \cite{fudenberg,WS00}. How to
motivate team building and cooperation in security and risk management has been recently discussed
in \cite{walid2} as well as in\cite[Chap. 6]{alpcan-book}. This alternative approach provides
a complementary and potentially very interesting research direction.

Some of the other open research directions follow directly from relaxing the
assumptions in Section~\ref{sec:model}. Improving the robustness of the 
incentive mechanisms against malicious units who do not follow the rules or
have utilities orthogonal to other users (sometimes referred to as adversarial mechanism
design) is an emerging and relevant research area. 
Detection of such misbehavior is also of both practical and theoretical interest.
In parallel to users, the relaxation of the assumption on risk manager's honesty leads
to similarly interesting questions such as how can a unit detect and respond to
misbehavior (e.g. unfairness) of the risk manager.



\section*{Appendix}

\subsection*{
Existence and Uniqueness of Nash Equilibrium} 

This appendix revisits the analysis in \cite{tansuphd,rosen} on existence
and uniqueness of Nash equilibrium.

In the strategic game $\mc G$ given in Definition~\ref{def:game}, the strategy (decision) space of the players
is assumed to be convex, compact, and has a nonempty interior. Furthermore, the cost functions of the players, 
$J_i, \;\; i \in \mc A$, is strictly convex in $x_i$ and at least twice continuously differentiable due to 
its definition as well as those of utility functions $U_i, \;\; i \in \mc A$. Therefore, the game $\mc G$ 
admits (at least) a Nash equilibrium from Theorem 4.4 in~\cite[p.176]{basargame}.

Next, additional conditions are imposed such that the game $\mc G$
admits a unique NE solution.  Toward this end, define the pseudo-gradient operator
\begin{equation}\label{e:psgrad}
 \overline \nabla J:= \left [\partial J_1(x) / \partial x_1 \cdots
  \partial J_N(x) / \partial x_N \right ]^T  := g(x).
\end{equation}
Subsequently, let the $N \times N$ matrix $G(x)$ be the Jacobian of $g(x)$
with respect to $x$:
\begin{equation}\label{e:g1}
 G(x):=
   \begin{pmatrix}
   b_1 & a_{12} & \cdots & a_{1N} \\
   \vdots &   & \ddots & \vdots \\
   a_{N1} & a_{N2}   & \cdots     & b_N \
 \end{pmatrix} \;,
\end{equation}
where $b_i$ and $a_{ij}$ are defined as
$b_i:=\frac{\partial^2 J_i(x)}{\partial x_i^2}$ and
$a_{i,j} :=  \frac{\partial^2 J_i(x)}{\partial x_i \partial x_j}$,
respectively. 
\begin{assm} \label{assm3}
The symmetric matrix $G(x)+G(x)^T$, where $G(x)$ is defined in~(\ref{e:g1}),
is positive definite, i.e. $G(x)+G(x)^T >0$ for all $x \in \mc X$.
\end{assm}

\begin{assm} \label{assm4}
The strategy space $\mc X$ of the game $\mc G$ can be described as 
\begin{equation} \label{e:xdef}
 \mc X:=\{x \in \Real^N\,:\, h_j(x) \leq 0 ,\; j=1,\,2,\,\ldots r\},
\end{equation} 
where $h_j: \Real^N \rightarrow \Real, j=1,\,2,\,\ldots r$, $h_j(x)$ is convex in its arguments
for all $j$, and the set $\mc X$ is bounded and has a nonempty interior. In addition, the derivative of 
at least one of the constraints with respect to $x_i$,
$\{d h_j(x) / d x_i ,\; j=1,\,2,\,\ldots r\}$, is nonzero for $i=1,\,2,\,\ldots N$, $\forall x \in \mc X$.
\end{assm}

Now, revisiting the analysis in \cite{tansuphd,rosen}, it is shown that the game  $\mc G$ admits
a unique Nash equilibrium under Assumptions~\ref{assm3} and \ref{assm4}.

In view of Assumption~\ref{assm4}, the Lagrangian function for player $i$ 
in this game is given by
\begin{equation} \label{e:lagrangian}
L_i(x,\mu)=J_i(x)+\sum_{j=1}^r \mu_{i, j} h_j(x) ,
\end{equation} 
where $\mu_{i, j}\geq 0,\;  j=1,\,2,\,\ldots r$ are the Lagrange multipliers of
player $i$~\cite[p. 278]{bertsekas2}.  
We now provide a proposition for the  game  $\mc G$   with conditions similar 
to the well known Karush-Kuhn-Tucker necessary conditions 
(Proposition 3.3.1, p. 310,~\cite{bertsekas2}).

\begin{prop} \label{kuhntucker}
Let $x^*$ be a NE point of the game $\mc G$ and Assumptions~\ref{assm3}-\ref{assm4}
hold. There exists then a unique set of Lagrange multipliers, 
$\{\phi_{i, j}:\;  j=1,\,2,\,\ldots r ,\; i=1,\,2,\,\ldots N \} $, such that
$$ \begin{array}{r}
 \dfrac{d L(x^*,\phi)}{d x_i}= \dfrac{d J_i(x^*)}{d x_i}+
 \displaystyle { \sum_{j=1}^r \phi_{i, j}^*\dfrac{d h_j(x^*)}{d x_i}} =0,\\ 
 i=1,\,2,\,\ldots N, \\
\phi_{i, j} \geq 0, \;\; \forall i, j,  \text{ and } \;
\phi_{i, j} = 0, \;\; \forall j \notin A_i(\x^*), \forall i\, ,
\end{array}
$$
where $A_i(x^*)$ is the set of active constraints in $i^{th}$ player's minimization
problem at NE point $x^*$.
\end{prop}

\begin{proof}
The proof essentially follows lines similar to the ones of the
Proposition 3.3.1 of~\cite{bertsekas2}, where the penalty approach
is used to approximate the original constrained problem by
an unconstrained problem that involves a violation of the
constraints. The main difference here is the
repetition of this process for each individual $x_i$ at the NE point $x^*$.
\end{proof}

Define now a more compact notation the vector of Lagrangian functions as 
$L:=[L_1,\ldots,L_N]$, and
the $N \times N$ diagonal matrix of Lagrange multipliers for the $j^{th}$ 
constraint as $\Phi_j=\rm{diag} [\phi_{1,j}, \phi_{2,j}, \ldots \phi_{N,j}]$. 


By Proposition~\ref{kuhntucker} and Assumption~\ref{assm4}, a NE point $x^{(1)}$ satisfies
\begin{equation} \label{e:proofunique1}
 \overline \nabla L(x^{(1)},\Phi^{(1)})=g(x^{(1)})+\sum_{j=1}^r \Phi_j^{(1)} 
\overline \nabla h_j(x^{(1)})=0,
\end{equation}
where $\Phi_j^{(1)}\geq 0$ is unique for each $j$.
Assume there are two different NE points $x^{(0)}$ and $x^{(1)}$. Then, one can also
write the counterpart of (\ref{e:proofunique1}) for $x^{(0)}$. 
Following an argument similar to the one in the proof of Theorem 2 in~\cite{rosen},
one can show that this leads to a contradiction. We present a brief outline of a simplified
version of that proof for the sake of completeness. 

Multiplying (\ref{e:proofunique1}) and its counterpart for $x^{(0)}$ from left by
$(x^{(0)}-x^{(1)})^T$, and then adding them together, we obtain
\begin{equation} \label{e:contradict}
\begin{array}{rcl}
 0 & = & (x^{(0)}-x^{(1)})^T \overline \nabla L(x^{(1)},\Phi^{(1)}) \\
  & & +  \left( \overline \nabla L(x^{(1)},\Phi^{(1)})\right)^T (x^{(0)}-x^{(1)}) \\
  & & + (x^{(1)}-x^{(0)})^T \overline \nabla L(x^{(0)},\Phi^{(0)})  \\ \\
  &= & (x^{(0)}-x^{(1)})^T \left( g(x^{(1)})-  g(x^{(0)}) \right) \\
  & & + \left( g(x^{(1)})-  g(x^{(0)}) \right)^T (x^{(0)}-x^{(1)}) \\
  & & +  (x^{(1)}-x^{(0)})^T \sum_{j=1}^r  [\Phi_j^{(1)} \overline \nabla h_j(x^{(1)}) \\
  & & - \Phi_j^{(0)} \overline \nabla h_j(x^{(0)})] .\\
\end{array}
\end{equation}

Define the strategy vector $x(\theta)$ as a
convex combination of the two equilibrium points $x^{(0)}\,,\,x^{(1)} $ :
$$
  x(\theta)=\theta x^{(1)}  + (1-\theta) x^{(0)}  ,
$$
where $0<\theta<1$. Take the derivative of $g(x(\theta))$ with respect to $\theta$,
\begin{equation}\label{e:gdiff}
  \dfrac{dg(x(\theta))}{d\theta}=G(x(\theta)) \frac{dx(\theta)}{d\theta}=G(x(\theta))(x^{(1)} -x^{(0)}),
\end{equation}
where $G(x)$ is defined in~(\ref{e:g1}). Integrating~(\ref{e:gdiff}) over
$\theta$ yields
\begin{equation}\label{e:gintegral}
  g(x^{(1)})-g(x^{(0)})=\left[\int_0^1 G(x(\theta)) d\theta \right](x^{(1)}-x^{(0)} ) .
\end{equation}
Multiplying (\ref{e:gintegral})
from left by $(x^{(1)}-x^{(0)} )^T$, the transpose of (\ref{e:gintegral}) from right by
$(x^{(1)}-x^{(0)} )$, and adding these two terms yields
\begin{equation}\label{e:gintegral2}
  (x^{(1)}-x^{(0)})^T \left[\int_0^1 G(x(\theta))+G^T(x(\theta)) d\theta \right](x^{(1)}-x^{(0)} ) .
\end{equation}
Since $G(x(\theta))+G^T(x(\theta)) $ is positive definite by Assumption~\ref{assm3}
and the sum of two positive definite matrices is positive definite,
the matrix $\bar G:=\int_0^1 G(x(\theta))+G^T(x(\theta))  d\theta$ is positive definite.

Similarly, we have
\begin{equation}\label{e:hdiff}
  \dfrac{d \overline \nabla h(x(\theta))}{d\theta}=H(x(\theta)) \frac{dx(\theta)}{d\theta}=H(x(\theta))(x^1-x^0),
\end{equation}
where $H(x)$ is the Jacobian of $\overline \nabla h(x)$ and positive definite due to convexity of $h(x)$
by definition. The third term in (\ref{e:contradict})
$$ \begin{array}{r}
(x^{(0)}-x^{(1)} )^T \sum_{j=1}^r [\Phi_j^{(0)}\overline  \nabla h_j(x^{(0)})-\Phi_j^{(1)} 
  \overline  \nabla h_j(x^{(1)})],
\end{array}
$$
is less than 
$$ \sum_{j=1}^r [\Phi_j^{(1)}-\Phi_j^{(0)}] [h_j(x^{(1)})-h_j(x^{(0)})] ,$$
due to convexity of $h(x)$. Since for each constraint $j$, $h_j(x) \leq 0\; \forall x$, $\Phi_j^{(i)} h_j(x^{(i)})=0,\; i=0, 1$,  and
$\Phi_j$ is positive definite, where the latter two follow from Karush-Kuhn-Tucker
conditions, this term is also non-positive.

The sum of the first two terms in (\ref{e:contradict}) are the negative of (\ref{e:gintegral2}), which
is strictly positive for all $x^{(1)} \neq x^{(0)}$. Hence, (\ref{e:contradict})
is strictly negative which leads to a contradiction unless $x^{(1)}=x^{(0)}$. Thus, there exists a unique NE point in the  game $\mc G$.

\bibliographystyle{IEEEtranS}

\end{document}